\documentclass[a4paper,UKenglish,cleveref, autoref, thm-restate]{lipics-v2021}

\title{Sensitivity of Repetitiveness Measures to String Reversal} %

\newcommand{\isct}{%
M\&D Data Science Center, Institute of Integrated Research, 
Institute of Science Tokyo, Japan
}

\newcommand{\kyushu}{%
Department of Informatics, Kyushu University, Japan
}

\author{Hideo Bannai}{\isct}
{hdbn.dsc@tmd.ac.jp}
{https://orcid.org/0000-0002-6856-5185}{JSPS KAKENHI Grant Number JP24K02899}

\author{Yuto Fujie}{Joint Graduate School of Mathematics for Innovation, Kyushu University, Japan}
{fujie.yuto.104@s.kyushu-u.ac.jp}
{}{}

\author{Peaker Guo}{\isct}
{peakerguo@gmail.com}
{https://orcid.org/0000-0002-9098-1783}{}

\author{Shunsuke Inenaga}{\kyushu}
{inenaga.shunsuke.380@m.kyushu-u.ac.jp}
{https://orcid.org/0000-0002-1833-010X}{JSPS KAKENHI Grant Numbers JP23K24808 and 23K18466}

\author{Yuto Nakashima}{\kyushu}
{nakashima.yuto.003@m.kyushu-u.ac.jp}
{https://orcid.org/0000-0001-6269-9353}{JSPS KAKENHI Grant Number JP25K00136}

\author{Simon J. Puglisi}{Department of Computer Science, University of Helsinki, Finland}
{simon.puglisi@helsinki.fi}
{https://orcid.org/0000-0001-7668-7636}{}

\author{Cristian Urbina}
{Faculty of Mathematics, Informatics and Mechanics, University of Warsaw,  Poland \and Center for Biotechnology and Bioengineering (CeBiB),  Chile}
{crurbina1997@gmail.com}
{https://orcid.org/0000-0001-8979-9055}{Polish National Science Center, grant no. 2022/46/E/ST6/00463; Basal Funds FB0001 and AFB240001, ANID, Chile; and FONDECYT Project 1-230755, ANID, Chile.}

\authorrunning{H. Bannai et al.} 

\Copyright{Jane Open Access and Joan R. Public} 



\category{} 

\relatedversion{} 

\acknowledgements{The authors thank Gonzalo Navarro for their helpful comments and discussion.}

\nolinenumbers 

\EventEditors{John Q. Open and Joan R. Access}
\EventNoEds{2}
\EventLongTitle{42nd Conference on Very Important Topics (CVIT 2016)}
\EventShortTitle{CVIT 2016}
\EventAcronym{CVIT}
\EventYear{2016}
\EventDate{December 24--27, 2016}
\EventLocation{Little Whinging, United Kingdom}
\EventLogo{}
\SeriesVolume{42}
\ArticleNo{23}

\usepackage[T1]{fontenc}
\usepackage{graphicx}
\usepackage{amsmath}

\usepackage{hyperref}
\usepackage{color}
\usepackage[dvipsnames]{xcolor}

\renewcommand{\a}{\mathtt{a}}
\renewcommand{\b}{\mathtt{b}}
\renewcommand{\c}{\mathtt{c}}
\renewcommand{\d}{\mathtt{d}}
\newcommand{\dol}{\text{\tt \$}}
\newcommand{\BWT}{{\tt BWT}}

\newcommand{\BBWT}{{\tt BBWT}}
\newcommand{\LF}{{\tt LF}}
\newcommand{\LZ}{{\tt LZ}}
\newcommand{\LEX}{{\tt LEX}}
\newcommand{\rdol}{r_{\dol}}
\newcommand{\rle}{{\tt rle}}

\newcommand{\lsym}{\text{\tt \#}}
\newcommand{\rsym}{\text{\tt \&}}
\newcommand{\asym}{\a}
\newcommand{\bsym}{\b}

\newcommand{\LCP}{\mathit{LCP}}
\newcommand{\ISA}{\mathit{ISA}}
\newcommand{\SA}{\mathit{SA}}

\usepackage{todonotes}

\begin{document}

\maketitle

\begin{abstract}We study the impact that string reversal can have on several repetitiveness measures. First, we exhibit an infinite family of strings where the number, $r$, of runs in the run-length encoding of the Burrows--Wheeler transform (BWT) can increase additively by $\Theta(n)$ when reversing the string. This substantially improves the known $\Omega(\log n)$ lower-bound  for the additive sensitivity of $r$ and it is asymptotically tight. We generalize our result to other variants of the BWT, including the variant with an appended end-of-string symbol and the bijective BWT. We show that an analogous result holds for the size $z$ of the Lempel--Ziv 77 (LZ) parsing of the text, and also for some of its variants, including the non-overlapping LZ parsing, and the LZ-end parsing. Moreover, we describe a family of strings for which the ratio $z(w^R)/z(w)$ approaches $3$ from below as $|w|\rightarrow \infty$. We also show an asymptotically tight lower-bound of $\Theta(n)$ for the additive sensitivity of the size $v$ of the smallest lexicographic parsing to string reversal. Finally, we show that the multiplicative sensitivity of $v$ to reversing the string is $\Theta(\log n)$, and this lower-bound is also tight. Overall, our results expose the limitations of repetitiveness measures that are widely used in practice, against string reversal---a simple and natural data transformation.
\end{abstract}

\clearpage
\setcounter{page}{1}

\section{Introduction}

In many fields of science and industry, there exist huge data collections. Many of these collections are highly repetitive, in the sense that the documents that make them up are highly similar to each other. Examples include the \emph{1000 Human Genomes Project}~\cite{1000/2015} in bioinformatics and the \emph{Software Heritage Repository}~\cite{CZ2017}. Repetitiveness measures~\cite{NavSurveyACM} were introduced to quantify the degree of compressibility of highly repetitive string collections, as other standard techniques based on Shannon's entropy fail to capture repetitiveness~\cite{KreftNavarro}. 

Some repetitiveness measures abstract the output size of existing compressors, like the size $z$ of the \emph{Lempel-Ziv 77 parsing (LZ)}~\cite{LZ1976}, the size $g$ of the smallest \emph{straight-line program (SLP)} generating the string~\cite{KY2000}, and the number $r$ of \emph{equal-symbol runs in the Burrows--Wheeler Transform (BWT)}~\cite{BW94}. Other measures are based purely on combinatorial properties of the strings, like the \emph{substring complexity} $\delta$~\cite{RRRS2013} and the size $\gamma$ of the smallest \emph{string attractor}~\cite{KP2018}.

A fundamental question in this context is what makes a repetitiveness measure $\mu(w)$ better than another. Clearly, one important aspect is \emph{space efficiency}, that is, how small $\mu(w)$ is on strings that are highly repetitive. Another important aspect is \emph{reachability}: whether it is possible to represent all strings within $O(\mu(w))$ space. There is also \emph{indexability}: what queries on the uncompressed string can be answered efficiently using $O(\mu)$ space. 

Recently, another aspect of repetitiveness measures, which we call \emph{robustness}, has gained attention. 
A measure is robust against a string operation if it does not increase much after applying that operation to any string. We are interested in quantifying how much a given measure can increase after applying a string operation in worst-case scenarios. 
This quantity is called \emph{sensitivity}~\cite{AFI23}.
Questions about the sensitivity of a measure are especially important in dynamic domains where the data can change over time, and finding answers to them can be useful, for instance, to understand how to update the measure after applying such a string operation. Akagi~et~al.~\cite{AFI23} studied how much repetitiveness measures can change after applying to the string single character edit operations. There are many other works dealing with similar problems for different combinations of string operation and repetitiveness measures~\cite{BCLR2025,ConstantinescuThesis,ConstantinescuIlie5,FRSU2025,GILPST21,GILRSU2025,JK2025,NKFIB2024,NOU2025,NRU2025,NRU2025arxiv,NU2025}. 

In this work we focus on studying the fundamental \emph{string reversal} operation, which takes a string $w[1]\cdots w[n]$ and outputs the string $w[n]\cdots w[1]$. 
We first review known results on the impact of string reversal on repetitiveness measures, and then present our contributions. \cref{table:sensitivity_reverse} summarizes both the prior results and our contributions. 

On the one hand, there are measures of repetitiveness that are \emph{fully symmetric}, that is, reversing the string does not change the value of the measure at all. This is a desirable property, as arguably, mirroring the string does not change its intrinsic repetitiveness. This class includes the measures $\delta$ and $\gamma$, the size $\nu$ of the smallest NU-system~\cite{NU2025}, the size $b$ of the smallest \emph{bidirectional macro scheme}~\cite{SS1982}; and the size $g,g_{rl},g_{it}$ and $c$ of the smallest  straight-line program, run-length straight-line program~\cite{Nishimoto2016}, iterated straight-line program~\cite{NOU2025}, and collage-system~\cite{KMSTSA2003}, respectively. A weaker version of symmetry is exhibited by the size $\chi$ of the smallest \emph{suffixient set}, for which it was recently proved that $\chi(w^R)/\chi(w)\le 2$~\cite{NRU2025arxiv}, though  $\chi(w^R)-\chi(w)$ can be as large as $\Omega(\sqrt{n})$~\cite{NRU2025}.

At the other extreme, there are measures widely used in practice---such as the number $r$ of runs in the RLBWT---that are very sensitive to reversal. Giuliani~et~al.~\cite{GILPST21} showed that in an infinite family of words, after reversing the string, the measure $r$ can increase from $O(1)$ to $\Omega(\log n)$.
An equivalent result was obtained by Biagi~et~al.~\cite{BCLR2025} regarding the number $r_B$ of runs in the run-length \emph{Bijective Burrows-Wheeler Transform (BBWT).} Moreover, the size $e$ of the \emph{compact directed acyclic word graph (CDAWG)} of a string can increase by a $\Omega(\sqrt{n})$ multiplicative factor when reversing the string~\cite{IK2025arxiv}. For the additive case, the increase can be up to $\Omega(n)$ (see Appendix \ref{sec:add-sen-e}).

For some other measures, not much is known. For LZ parsing and some of its variants, it has been conjectured that even though the number of phrases of the parses can change after reversing the string, the ratio $z(w^R)/z(w)$ is bounded by a constant. For the size $v$ of the \emph{lex-parse} of the string, the conjecture is the opposite, namely that $v(w^R)/v(w)$ can be $\omega(1)$.

In this context, our contributions are as follows:
\begin{enumerate}
\item In \cref{sec:add-sen-rlbwt}, we show that the number $r$ (resp.\ $r_B$) of runs of the RLBWT (resp.\ RLBBWT) can increase by $\Theta(n)$ after reversing the string. This improves the currently known $\Omega(\log n)$ lower bound for the additive increase on these measures under string reversal~\cite{BCLR2025,GILPST21}. 
\item In \cref{sec:sen-lempel-ziv-parse}, we show that the number $z$ (resp.\ $z_{no}$, resp.\ $z_e$) of phrases of the Lempel-Ziv parsing (resp.\ LZ without overlaps, resp.\ LZ-end) can increase by $\Theta(n)$ when applying the reverse operation. This improves the currently known $\Omega(\sqrt n)$  lower bound for the additive increase on $z$~\cite{ConstantinescuThesis}. Moreover, we show a family where  $z(w^R)/z(w)$ approaches~$3$ as $|w|$ goes to infinity. 
\item Finally, in \cref{sec:sen-lex-parse}, we show that the number $v$ of phrases in the lex-parse  of a string can grow by a $\Omega(\log n)$ multiplicative factor after applying the reverse operation, and this is asymptotically tight. Moreover, we show a string family where after reversing the string, the additive increase on the measure $v$ is $\Theta(n)$.
\end{enumerate}
We conclude in \cref{sec:conclusions} with conclusions and future directions for research.

\begin{table}
\caption{
Summary of the best known bounds on the multiplicative and additive sensitivity of repetitiveness measures to string reversal. While the upper-bounds hold for all strings, the lower-bounds hold for specific infinite string families.
We mark results obtained in this work with $\dagger$.
}
\label{table:sensitivity_reverse}
\centering
\setlength{\tabcolsep}{4pt}
\begin{tabular}{|c|c|c|c|c|}
\hline
Measure & MS lower-bound & MS upper-bound & AS lower-bound & AS upper-bound\\ \hline
$r, \rdol$ & $\Omega(\log n)$~\cite{GILPST21}, \textsuperscript{$\dagger$} & $O(\log^2n)$~\cite{RomanaThesis, KK2022} & $\Omega(n)^{\dagger}$ & $O(n)$\\ \hline
$r_B$ & $\Omega(\log n)$~\cite{BCLR2025} & $O(\log^3 n)^{\dagger}$& $\Omega(n)^{\dagger}$ & $O(n)$ \\ \hline
$z,z_{no}$ & $3-o(1)^{\dagger}$ & $O(\log n)^{\dagger}$ & $\Omega(n)^{\dagger}$  & $O(n)$\\ \hline
$z_{end}, z_e$ & $1$ & $O(\log n)^{\dagger},O(\log^2 n)^{\dagger}$ & $\Omega(n)^{\dagger}$  & $O(n)$\\ \hline
$v$ & $\Omega(\log n)^{\dagger}$  & $O(\log n)^{\dagger}$ & $\Omega(n)^{\dagger}$  & $O(n)$\\ \hline
$\chi$ & $1$  & $2$~\cite{NRU2025arxiv} & $\Omega(\sqrt{n})$~\cite{NRU2025} & $O(n)$ \\ \hline
$e$ & $\Omega(\sqrt{n})$~\cite{IK2025arxiv}& $O(\sqrt{n})$~\cite{IK2025arxiv}& $\Omega(n)^\dagger$ & $O(n)$\\ \hline
$\delta,\gamma,\nu,b,c,g_{it},g_{rl}, g$ & $1$ & $1$ & $0$ & $0$\\ \hline
\end{tabular}
\end{table}

\section{Preliminaries}\label{sec:preliminaries}

We denote $[i,j]=\{i,i+1, \ldots,j\}\subseteq\mathbb{Z}$ if $i \le j$, and $\emptyset$ otherwise. An \emph{ordered alphabet} $\Sigma = \{a_1,a_2,\dots,a_\sigma\}$ is a finite set of symbols equipped with an order relation $<$  such that $a_1 < a_2 < \dots < a_\sigma$. 
A \emph{string} $w[1 \ldots n]=w[1] \cdots w[n]$ is a finite sequence of symbols in $\Sigma$. 
The \emph{length} of $w$ is denoted by $|w| = n$. 
The unique string of length $0$, called the \emph{empty string}, is denoted by $\varepsilon$. The set of all strings is $\Sigma^*$, and we let $\Sigma^+ = \Sigma^* \setminus \{\varepsilon\}$. 

For strings $x=x[1]\cdots x[n]$ and $y=y[1]\cdots y[m]$, their \emph{concatenation} is $x\cdot y \equiv xy = x[1]\cdots x[n]y[1]\cdots y[m]$. 
By $\prod_{i=1}^kw_i$ we denote the concatenation of the strings $w_1,w_2, \dots, w_k$  in that order. 
We also let $w^k = \prod_{i=1}^kw$. 
For a string $w$, a sequence of non-empty strings $(x_k )^m_{k=1} = (x_1,\dots,x_m)$ is referred to as a \emph{factorization} (or \emph{parsing}) of $w$ if $w = \prod_{i=1}^m x_i$; each $x_k$ is called a \emph{phrase} (or \emph{factor}).
The \emph{run-length encoding} of a string $w$ is $\rle(w)=(a_1,p_1),\ldots,(a_m,p_m)$ where $a_i \neq a_{i+1}$ for $i \in [1,m-1]$, $p_i\ge 1$ for all $i$, and $w=\prod_{i=1}^ma_i^{p_i}$. 
For a string $w = xyz$, the string $x$ (resp.\ $y$, resp.\ $z$) is called a  \emph{prefix} (resp.\ \emph{substring}, resp.\  \emph{suffix}) of $w$. 
A substring (resp.\ prefix, resp.\ suffix) is \emph{proper} if it is different from $w$, and \emph{non-trivial} if in addition it is different from $\varepsilon$. 
We use the notation $w[i \ldots j]$ to denote the string $w[i]\cdots w[j]$ when $1\le i \le j \le n$ and $\varepsilon$ otherwise. 
We denote $S_w(k)$ the set of substrings of length $k$ appearing in $w$.

The \emph{lexicographic order} on $\Sigma^*$ is defined recursively by $\varepsilon < x$ for every $x \in \Sigma^+$, and 
for $a,b \in \Sigma$ and $u,v \in \Sigma^*$, by $au < bv$ if either $a<b$, or $a=b$ and $u<v$. The \emph{$\omega$-order} is defined as $x<_\omega y$ iff $x^\omega < y^\omega$, where $x^\omega$ is the infinite periodic string with period $x$.

The \emph{suffix array} (SA) of a string $w$ of length $n$ is an array $\SA[1 \ldots n]$ such that $\SA[i]$ is the starting position of the $i^{\text{th}}$ lexicographically smallest suffix of $w$.
The \emph{inverse suffix array} (ISA) of $w$ is the array $\ISA[1 \ldots n]$ defined by $\ISA[\SA[i]] = i$ for $i \in [1, n]$.
The \emph{longest common prefix array} (LCP) of $w$ is an array $\LCP[1 \ldots n]$ where $\LCP[i]$ stores the length of the longest common prefix of the suffixes $w[\SA[i-1] \ldots n]$ and $w[\SA[i] \ldots n]$ for $i \in [2, n]$, and $\LCP[1]=0$.

A \emph{rotation} of $w[1\ldots n]$ is any string of the form $w[i\ldots n]w[1\ldots i-1]$ for $i \in [1,n]$. Let $\mathcal{R}(w)=\{w_1,w_2,\dots,w_n\}$ be the multiset of sorted rotations of $w$. The \emph{Burrows-Wheeler transform (BWT)} of $w$ is the string $\BWT(w)=w_1[n]w_2[n]\cdots w_n[n]$.  Moreover, let $\BWT_x(w)$ denote the substring of $\BWT(w)$ induced by the rotations prefixed by a string $x$. That is, let $[\mathit{lb}, \mathit{ub}]$ be the maximal range such that $w_i[1 \ldots |x|] = x$ for all $i \in [\mathit{lb}, \mathit{ub}]$; then $\BWT_x(w) = w_{\mathit{lb}}[n] \cdots w_{\mathit{ub}}[n]$. The \emph{BWT-matrix} of $w$ is a 2D array containing in the $i^{\text{th}}$ row the rotation  $w_i$ for $i \in [1,n]$; its last column coincides with $\BWT(w)$.

A \emph{Lyndon word} is any string $w$ that is smaller than any of its other rotations, or alternatively, smaller than any of its proper suffixes. The \emph{Lyndon factorization} of a string $w$ is the unique sequence $\LF(w)=x_1, \dots, x_\ell$ satisfying that $w=\prod_{i=1}^\ell x_i$, each $x_i$ is Lyndon, and $x_i \ge x_{i+1}$ for $i \in[1,\ell-1]$.

\subsection{Repetitiveness Measures}

We explain the repetitiveness measures relevant to our work. For an in-depth overview on repetitiveness measures, see Navarro’s survey~\cite{NavSurveyACM,NavSurveyArxiv}. We start by introducing two measures that act as lower-bounds for the other measures considered in this work.
One is the \emph{substring complexity} $\delta$ of a string $w$ \cite{KNP2023,RRRS2013}, defined as $\delta(w) = \max_{k\in[1,|w|]} |S_w(k)| / k$. The other one is $\gamma$ \cite{KP2018}, which we do not explain here, as we only use a few known results  on this measure.

A popular measure of repetitiveness is the size of the \emph{run-length encoding of the BWT (RLBWT)} of a string. The BWT is known for its indexing functionality~\cite{r-index} and indexing-related applications in Bioinformatics~\cite{soap2,Bowtie,bwa}.

\begin{definition}
We denote $r(w) = |\rle(\BWT(w))|$.
\end{definition}

Sometimes, a $\dol$-symbol, smaller than any other symbol, is assumed to appear at the end. We denote this variant $r_{\dol}(w) = |\rle(\BWT(w\dol))|$.

The \emph{Bijective Burrows--Wheeler Transform (BBWT)} of $w$, denoted $\BBWT(w)$, is obtained by considering the rotations of all the factors in the Lyndon factorization $\LF(w)$, then sorting them in $\omega$-order, and finally concatenating their last characters. By taking the \emph{run-length encoding of the BBWT (RLBBWT)} we obtain the following repetitiveness measure.

\begin{definition}We denote $r_B(w) = |\rle(\BBWT(w))|$.
\end{definition}

\begin{example}
    Let $w=\b\a\a\b\a\a\b\a$. Then, $\LF(w)=\b,\a\a\b,\a\a\b,\a$. We sort the rotations of $\b$, $\a\a\b$ (two times), and $\a$, in $\omega$-order and take their last characters:
    $$\a<_\omega\a\a\b\le_\omega\a\a\b<_\omega\a\b\a\le_\omega\a\b\a<_\omega\b\a\a\le_\omega\b\a\a<_\omega\b.$$
    In this case, $\BBWT(w)=\a\b\b\a\a\a\a\b$, $\rle(\BBWT(w))=(\a,1), (\b,2),(\a,4),(\b,1)$, and $r_B(w)=4$. It can be verified that $\BWT(w)= \b\b\b\a\a\a\a\a$ and $r(w)=2$.
\end{example}

The following measure is known for being one of the smallest reachable measures in terms of space, that is also computable in linear time~\cite{LZ1976}.

\begin{definition}
The \emph{Lempel-Ziv} parse of a string $w$ is a factorization of $w$ into phrases $\LZ(w) = x_1,\ldots,x_z$ 
such that, for each  $1 \leq j \leq z$, phrase $x_j$ is the longest prefix of $x_j\dots x_z$ with another occurrence in $w$ starting at position $i\le|x_1x_2\dots x_{j-1}|$; if no such prefix exists, then $x_j$ is the first occurrence of some character and $|x_j|=1$.
\end{definition}

We consider some variants of the Lempel-Ziv parsing and the measure $z$. The $z_{no}$ variant additionally requires that the source of any phrase must not overlap the phrase. The $z_e$ variant requires that the source of any phrase must have an occurrence aligned with the end of a previous phrase. While $z$ and $z_{no}$ are greedy and optimal among parses satisfying their respective conditions, $z_e$ is greedy, but not optimal: there can be other parsings smaller than $z_e$ such that each phrase has its source starting to its left and is aligned with a previous phrase. We denote $z_{end}$ the actual smallest of these parsings.

The last measure we consider is based on another type of \emph{ordered parsings}~\cite{NOP2021}. In the Lempel-Ziv parse, the source of a phrase appears before the phrase in $w$, in left-to-right order. In \emph{lexicographic parsings}~\cite{NOP2021}, the suffix starting at a given phrase must be greater than the suffix starting at its source, in lexicographic order. The size of left-to-right greedy version of this parsing is optimal, and is also considered a measure of repetitiveness.

\begin{definition}\label{def:lex-parse}
The \emph{lex-parse} of a string $w$ is a factorization of $w$ into phrases $\LEX(w) = x_1,\ldots,x_v$, 
such that, for each $1 \leq j \leq v$, phrase $x_j$ has
starting position $i = 1 + \sum_{t<j} |x_t|$ and 
length $\max\{1,\LCP[i']\}$, 
where $i' = \ISA[i]$. 
\end{definition}

\begin{example}Consider the string $w = \mathtt{abracadabracabra}$. Its Lempel-Ziv parse is $\LZ(w)=(\a,\b,\mathtt{r},\a,\c,\a,\d,\mathtt{abraca},\mathtt{bra})$. Its lex-parse is $\LEX(w)=(\mathtt{abraca},\d,\mathtt{abra},\c,\a,\b,\mathtt{r},\a)$. Thus, $z(w) =9$ and $v(w) = 8$.
\end{example}

\noindent
We remark that although the lex-parse can be computed from its definition in a left-to-right fashion, the source of each phrase can appear either to its left or to its right in $w$.

\subsection{Reverse Operation and Its Impact on Repetitiveness Measures}

The \emph{reverse} of a string $w=w[1]w[2]\cdots w[n]$ is the string $w^R = w[n]\cdots w[2]w[1]$. 
In what follows, we analyze the additive and multiplicative increase that reversing the string can have on the values of various measures. 
We consider worst-case scenarios, and we refer to those increases as \emph{additive sensitivity} and \emph{multiplicative sensitivity} of a measure to the reverse operation. When considering upper-bounds for the sensitivity of a measure, we need them to hold for all strings. Lower-bounds generally hold for a specific infinite string family.
Some measures, such as $\delta$ and $\gamma$, are invariant to string reversal. We focus on those that are not. 

One of the repetitiveness measures for which the sensitivity to reverse has been studied is $z$. Constantinescu~\cite{ConstantinescuThesis} showed that $z$ can increase by $\Omega(\sqrt{n})$ when reversing the string. It is also conjectured that $z(w^R)/z(w)=O(1)$. For the upper-bound, it is known that $\gamma\le z \le z_{no}\le z_{end} = O(\gamma \log n)$~\cite{NavSurveyArxiv}, and $\gamma$ is invariant upon reversal, hence, for $z,z_{no}$, and $z_{end}$, the multiplicative sensitivity is $O(\log n)$. Recently, Kempa and Saha proved that $z_e=O(\delta \log^2n)$~\cite{KempaSaha}, and $\delta$ is invariant to reversals, hence, the multiplicative sensitivity of $z_e$ is $O(\log^2 n)$.

Giuliani~et~al.~\cite{GILPST21} showed that $r(w^R)/r(w)=\Omega(\log n)$. Specifically, they describe a family where $r(w)=O(1)$ and $r(w^R)=\Theta(\log n)$, hence, in this family also holds $r(w^R)-r(w) = \Omega(\log n)$. This result can be extended to $\rdol$ (see Appendix~\ref{sec:mul-sen-rdol}).  Kempa and Kociumaka~\cite{KK2022} showed that $\rdol = O(\delta\log^2(n))$, and because $\rdol = \Omega(\delta)$ and $\delta$ is invariant to reversals, it follows that $\rdol(w^R)/\rdol(w) =  O(\log^2 n)$. This result was also extended to $r$~\cite{RomanaThesis}. 

Recently, Biagi~et~al.~\cite{BCLR2025} showed a family  where $r_B(w)=O(1)$ and $r_B(w^R)=\Theta(\log n)$. For an upper-bound, we can obtain a weaker version of the upper-bound for $r$: because $r_B = \Omega(\delta)$, $z = O(\delta\log n )$, and $r_B = O(z\log^2n)$~\cite{BBK2024}, it follows that $r_B(w^R)/r_B(w) = O(\log^3n)$. 

Regarding lexicographic parsings, it can be derived from the relation $\delta \le v = O(\delta \log n)$ and the fact that $\delta$ is invariant upon reversal, that $v(w^R)/v(w)=O(\log n)$~\cite{KNP2023,NOP2021}. 

\section{Additive Sensitivity of RLBWT Variants to Reverse}\label{sec:add-sen-rlbwt}

In this section, we show that under string reversal, the additive sensitivity of $r$, $ \rdol$, and $r_B$ are all $\Omega(n)$. This substantially improves the previous $\Omega(\log n)$ lower-bounds on the additive sensitivity, which can be derived from the results of Giuliani~et~al.~\cite{GILPST21} and Biagi~et~al.~\cite{BCLR2025} on the multiplicative sensitivity of $r$ and $r_B$. 

In the following proofs, we rely on an infinite family of strings defined as follows.

\begin{definition}[$u_k$]
Let $\Sigma = \{\asym,\bsym\} \cup \bigcup_{i\in[1,k]} \{\lsym_i,\rsym_i\}$ such that 
$$\lsym_1< \rsym_1  <  \lsym_2  <  \rsym_2  <  \dots  <  \lsym_k  <  \rsym_k   <  \asym  <  \bsym.$$ 
We let $u_k = \prod_{i=1}^k \bsym \asym \lsym_i \asym \rsym_i$. It holds $|u_k|=5k$.
\end{definition}

\subsection{Additive Sensitivity of \texorpdfstring{$r$}{} and \texorpdfstring{$\rdol$}{}}

We first fully characterize the BWT of strings $u_k$ and their reverse $u_k^R$. Then, we count the number of BWT-runs they have and how they differ. \cref{fig:BWT_uk} gives an example for $u_3$.

\begin{figure}[t]
    \centering
    \includegraphics[width=0.7\linewidth]{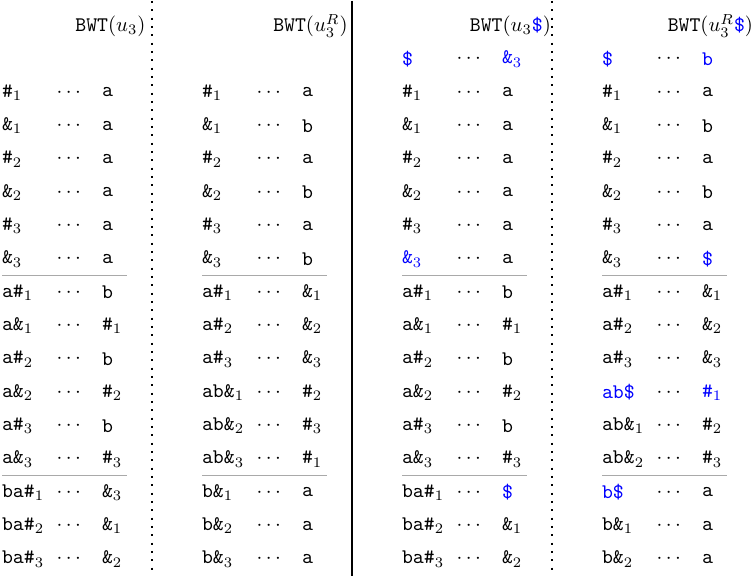}
    \caption{
    Illustration of \cref{le:r_uk}, \cref{le:r_ukR}, and \cref{prop:add-sen-r-dol} on
    $u_3= \bsym \asym \lsym_1 \asym \rsym_1  \bsym \asym \lsym_2 \asym \rsym_2 \bsym \asym \lsym_3 \asym \rsym_3$.
    The BWT matrices of relevant strings are illustrated,
    with the prefixes of the rotations shown alongside the BWTs.
    The changes caused by appending $\dol$ are highlighted in blue.
    }
    \label{fig:BWT_uk}
\end{figure}

\begin{lemma}\label{le:r_uk}
$\BWT(u_k) = \asym^{2k}(\prod_{i=1}^k\bsym\lsym_i) \rsym_k (\prod_{i=1}^{k-1} \rsym_{i})$ and $r(u_k)=3k+1$.
\end{lemma}

\begin{proof}
We characterize the BWT of $u_k$ and count its number of runs. 
Observe that in $u_k$, the rotations starting with the symbols $\lsym_i$ and $\rsym_i$ for $i \in [1,k]$ are always preceded by $\asym$. 
The rotations starting with $\bsym$'s are sorted according to the following symbols $\asym\lsym_i$ (or equivalently, according to their starting position within $u_k$, from left to right), and preceded by $\rsym_{i-1}$ for $i > 1$, and by $\rsym_k$ in the case of $i=1$ (i.e., $\bsym\asym\lsym_1\cdots\rsym_k$). 
The rotations starting with $\asym$ are sorted according to their following symbol $\lsym_i$ or $\rsym_i$ (or equivalently, according to their  starting position within $u_k$, from left to right), and preceded in an alternating fashion by $\bsym$, $\lsym_1$, $\bsym$, $\lsym_2$, \dots, $\bsym$, $\lsym_k$. 
Hence, we can partition the BWT of $u_k$ in the following disjoint ranges:
\begin{enumerate}[(1)]
\item\label{range:hash-dol-uk} $\BWT_{\lsym_i}(u_k) = \asym$ and $\BWT_{\rsym_i}(u_k)  = \asym$ for $i \in [1,k]$.
\item $\BWT_{\asym}(u_k)  = \prod_{i=1}^k\bsym\lsym_i$.
\item $\BWT_{\bsym}(u_k)  = \rsym_k (\prod_{i=1}^{k-1} \rsym_{i})$.
\end{enumerate}
Though interleaved, the ranges in~(\ref{range:hash-dol-uk}) yield a prefix of the BWT of $u_k$ of length $2k$ formed by only $\asym$'s. Thus, $\BWT(u_k) = \asym^{2k}(\prod_{i=1}^k\bsym\lsym_i) \rsym_k (\prod_{i=1}^{k-1} \rsym_{i})$ and $r(u_k)=3k+1$.
\end{proof}

\begin{lemma}\label{le:r_ukR}
$\BWT(u_k^R) = (\asym\bsym)^{k}(\prod_{i=1}^k\rsym_i) (\prod_{i=2}^k\lsym_i) \lsym_1 \asym^k$ and $r(u_k^R)=4k+1$.
\end{lemma}

\begin{proof}
First observe that $u_k^R= \prod_{j=0}^{k-1} \rsym_{k-j} \asym \lsym_{k-j} \asym \bsym$,
and in the following, we have $i = k-j$.
We characterize the BWT of $u_k^R$ and count its number of runs. Observe that just as for $u_k$, in $u_k^R$ the rotations starting with $\lsym_i$ for $i \in [1,k]$ are always preceded by $\asym$. On the other hand, the rotations starting with $\rsym_i$ are always preceded by $\bsym$. The rotations starting with $\bsym$'s are always preceded by $\asym$. The rotations starting with $\asym\lsym_i$  for $i \in [1,k]$ are sorted according to $\lsym_i$ (or equivalently, according to their starting positions within $u_k^R$ from right to left), and preceded by $\rsym_1$,\dots, $\rsym_k$. 
The rotations starting with $\asym\bsym$ are sorted according to their following symbol $\rsym_i$, and preceded by $\lsym_{i+1}$ when $i < k$ and $\lsym_1$ when $i = k$. Hence, we can partition the BWT of $u_k^R$ in the following disjoint ranges.
\begin{enumerate}[(1)]
\item\label{range:hash-dol-uk-r} $\BWT_{\lsym_i}(u_k^R) = \asym$ and $\BWT_{\rsym_i}(u_k^R) = \bsym$ for $i \in [1,k]$.
\item $\BWT_{\asym\lsym_i}(u_k^R) = \rsym_i$ for $i \in [1,k]$.
\item $\BWT_{\asym\bsym}(u_k^R) = (\prod_{i=2}^k\lsym_i) \lsym_1$.
\item $\BWT_{\bsym}(u_k^R) = \asym^k$.
\end{enumerate}
The interleaved ranges in~(\ref{range:hash-dol-uk-r}) yield a prefix of the BWT of $u_k^R$ of the form $(\asym\bsym)^k$.
Thus, $\BWT(u_k^R) = (\asym\bsym)^{k}(\prod_{i=1}^k\rsym_i) (\prod_{i=2}^k\lsym_i) \lsym_1 \asym^k$ and $r(u_k^R)=4k+1$.
\end{proof}

Recall that $|u_k|=5k$. By \cref{le:r_uk} and \cref{le:r_ukR} we derive the following result.

\begin{proposition}
There exists an infinite string family where $r(w^R)-r(w)=n/5=\Theta(n)$.
\end{proposition}

For the $\rdol$ variant, we observe that the BWT of $u_k$ does not change much after appending a $\dol$ at its end, because the relative order among all rotations of $u_k\dol$ is decided before comparing against the $\dol$ (except for $\dol u_k$).  In the case of $u_k^R$, the two rotations of $u_k^R$ starting with $\asym\bsym\rsym_k$ and $\bsym \rsym_k$ are changed so they start with $\asym\bsym\dol$ and $\bsym \dol$, respectively. For both $\rdol(u_k)$ and $\rdol(u_k^R)$, the number of runs increases by exactly one. The differences in their BWTs with respect to the version without the $\dol$ are highlighted in \cref{fig:BWT_uk}. We obtain the following.

\begin{proposition}\label{prop:add-sen-r-dol}
There exists a string family where $r_\dol(w^R)-r_\dol(w) = n/5 = \Theta(n)$.
\end{proposition}

\subsection{Additive Sensitivity of \texorpdfstring{$r_B$}{}}

We first characterize the Lyndon factorization of $u_k$ and $u_k^R$.

\begin{lemma}$\LF(u_k) = (\bsym,\asym,\lsym_1\asym\rsym_1\prod_{i=2}^k\bsym\asym\lsym_i\asym\rsym_i)$ and $\LF(u_k^R) 
= \left( \rsym_{k-j} \asym, \lsym_{k-j}\asym\bsym\right)_{j=0}^{k-1}  
$.
\end{lemma}

\begin{proof}
The Lyndon factorization of any string is unique. The factors in the proposed factorization of $u_k$ are indeed Lyndon, and appear in non-increasing order. Regarding $u_k^R$, it follows from the order $\lsym_1<\rsym_1 <\dots<\lsym_k<\rsym_k$ that the factors in the factorization are sorted in non-increasing order. All of these are Lyndon because each one of them starts with the unique smallest symbol of the factor.
\end{proof}

\begin{lemma}\label{le:rb_uk}
$\BBWT(u_k)= \rsym_k\asym^{2k-1}\lsym_1(\prod_{i=2}^k\bsym\lsym_i)\asym
(\prod_{i=1}^{k-1} \rsym_{i})
\bsym$ 
and $r_B(u_k) = 3k+2$.
\end{lemma}

\begin{proof}Let $u_k'= \lsym_1\asym\rsym_1\prod_{i=2}^k\bsym\asym\lsym_i\asym\rsym_i$ be the third Lyndon factor of $\LF(u_k)$.  First, we show that $\BWT(u_k') =  \rsym_k\asym^{2k-1}\lsym_1(\prod_{i=2}^k\bsym\lsym_i)(\rsym_1\cdots\rsym_{k-1})$. There are only a few changes in the BWT-matrix of $u_k'$ in comparison with $u_k$:  1) the rotation starting with $\bsym$ and last character $\rsym_k$ is lost; 2) the rotation starting with $\asym\lsym_1$ and preceded by $\bsym$ is lost; 3) the rotation starting with $\lsym_1$ is preceded by $\rsym_k$ instead of $\asym$ and remains the first rotation. All the other rotations keep their preceding symbol, and also their relative order, because the order of all rotations of $u_k$ coincides with their suffix order. The Lyndon factor $\asym$ is ordered after the rotation $\asym\rsym_k\dots\lsym_k$ and the first rotation starting with $\asym\bsym$. The Lyndon factor $\bsym$ is the greatest rotation of any Lyndon factor in $\LF(u_k)$ in $\omega$-order. Thus, $\BBWT(u_k)= \rsym_k\asym^{2k-1}\lsym_1(\prod_{i=2}^k\bsym\lsym_i)\asym(\rsym_1\cdots\rsym_{k-1})\bsym$ and $r_B(u_k) = 3k+2$.
\end{proof}

\begin{lemma}\label{le:rb_ukR}
$\BBWT(u_k^R) = (\bsym\asym)^k(\prod_{i=1}^k\rsym_i)(\prod_{i=1}^k\lsym_i)\asym^k$ and $r_B(u_k^R) = 4k+1$
\end{lemma}

\begin{proof}
The rotations of the Lyndon factors of $u_k^R$ are sorted according to their $\omega$-order, as follows:
%
$\lsym_i\asym\bsym$ followed by $\rsym_i\asym$ for $i \in [1,k]$;
$\asym\rsym_i$ for $i \in [1,k]$;
$\asym\bsym\lsym_i$ for $i \in [1,k]$;
$\bsym\lsym_i\asym$ for $i \in [1,k]$.
%
Thus, $\BBWT(u_k^R) = (\bsym\asym)^k(\prod_{i=1}^k\rsym_i)(\prod_{i=1}^k\lsym_i)\asym^k$ and $r_B(u_k^R) = 4k+1$.
\end{proof}

Remind that $|u_k|=5k$. By \cref{le:rb_uk} and \cref{le:rb_ukR} we derive the following result.

\begin{proposition}
There exists a string family where $r_B(w^R)-r_B(w) = n/5-1=\Theta(n)$.
\end{proposition}

\section{Sensitivity of the Lempel-Ziv Parsing}\label{sec:sen-lempel-ziv-parse}

\subsection{Multiplicative Sensitivity}

Experimentation~\cite{kkp2014} suggests that the multiplicative sensitivity of the Lempel-Ziv parsing to the reverse operation is $O(1)$, although the exact constant remains unclear. 
In~\cite{kkp2014}, an essentially brute-force search over short strings revealed a binary string 
$$T = \a\b\a\b\a\b\a\b\a\b\a\a\b\a\b\a\b\a\a\b\a\b\a\a\a\b\a\b\a\a\a\b\a\b\a\a\a\b\a\b\b\a\b\a\b\a\a\b\a\b\a\a\a\b\a$$ 
of length 55 with $z(T)=14$ and $z(T^R) = 6$.

We now show a family for which $z(w^R)/z(w)$ approaches $3$ from below as $|w|\rightarrow \infty$.

\begin{definition}[$T_p$]
For $p \geq 1$, let 
$
\Sigma = \{ x \} \cup \bigcup_{i\in[1, p]} \{ a_i, b_i \}
$.
For $i \in [1, p]$, define 
$A_i = a_i \cdots  a_1$ and 
$B_i = b_i \cdots  b_1$.
Then define
$\mathcal{A}_p = A_p \cdots A_1$, and  
$\mathcal{B}_p = B_1 \cdots B_p$. 
Next, let
$m_p = |\mathcal{A}_p| = |\mathcal{B}_p| = p(p+1)/2$,
and for $j \in [1, m_p]$, define 
\begin{align*}
G_j &= 
    \begin{cases}
    \mathcal{A}_p[m_p-j+1 \ldots m_p] \cdot  x \cdot  \mathcal{B}_p[1 \ldots j+1] & \text{for~}  1 \leq j \leq m_p-1, \\
    \mathcal{A}_p \cdot  x \cdot  \mathcal{B}_p & \text{for~} j = m_p.
    \end{cases}
\end{align*}
Finally, we define $T_p =  \prod_{j=1}^{m_p} G_j$,
noting that $|T_p| \in \Theta(p^4)$.
\end{definition}
\begin{figure}[t]
    \centering
    \includegraphics[width=\linewidth]{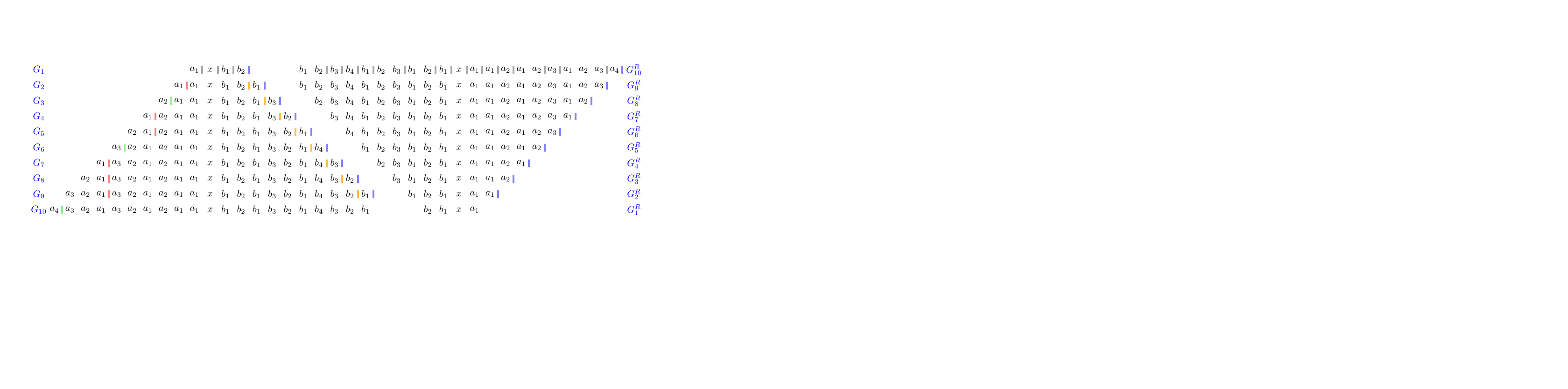}
    \caption{
    Illustration of $T_p = G_1 \cdots G_{m_p}$ (left) and $T_p^R = G_{m_p}^R \cdots G_1^R$ (right) for $p=4$.
    Note that
    $\mathcal{A}_4 = a_4 a_3 a_2 a_1 \cdot a_3 a_2 a_1 \cdot a_2 a_1 \cdot a_1$, 
    $\mathcal{B}_4 = b_1 \cdot b_2 b_1 \cdot b_3 b_2 b_1 \cdot b_4 b_3 b_2 b_1$, and $m_4  = 10$.
    The colored boundaries illustrate \cref{lem:lz-t-p} and \cref{lem:lz-t-p-r}.
    }
    \label{fig:lz-t-p}
\end{figure}

\begin{observation}\label{obs:lz-A-p-B-p}
$z(\mathcal{A}_p) = z(\mathcal{B}_p) = z(\mathcal{A}_p^R) = z(\mathcal{B}_p^R) = 2p-1$.
\end{observation}

\begin{proof}
The Lempel-Ziv factorization of $\mathcal{A}_p$ is
$\left( (a_{p-u})_{u=0}^{p-1}, \left( A_{p-v} \right)_{v=1}^{p-1} \right)$,
where
$(a_{p-u})_{u=0}^{p-1} = (a_p, \ldots, a_1)$ and 
$ 
( A_{p-v} )_{v=1}^{p-1}
= (  A_{p-1}, \ldots, A_1 )$, 
which means $z(\mathcal{A}_p) = p + (p-1) = 2p-1$.
The Lempel-Ziv factorization of $\mathcal{B}_p$ is
$\left( b_1, \left( b_{u+1}, B_{u} \right)_{u=1}^{p-1} \right) 
= (b_1, b_2, B_1, b_3, B_2, \ldots, b_{p}, B_{p-1})$, 
so $z(\mathcal{B}_p) = 1 + 2(p-1) = 2p-1$.
We can similarly show that 
$z(\mathcal{A}_p^R) = z(\mathcal{B}_p^R) = 2p-1$.
\end{proof}

\begin{observation}\label{obs:uniq-b-j-a-k}
For any $i, k \in [1, p]$, if $b_i a_k$ occurs in $T_p$, then it occurs exactly once.
\end{observation}

\begin{proof}
Any occurrence of $b_i a_k$ must be of the form
$G_{j}[|G_j|] \cdot G_{j+1}[1]$ for some $j \in [1, m_p-1]$.
By construction, $G_{j}[|G_j|]$ and $G_{j+1}[1]$ correspond to the following two sequences, respectively:
\begin{align*}
 \left( (b_{v-u+1})_{u=1}^{v} \right)_{v=2}^{p} 
 &= (b_2, b_1), (b_3, b_2, b_1), (b_4, b_3, b_2, b_1), \ldots, (b_{p}, b_{p-1}, \ldots, b_1), \text{~and} \\
 \left( (a_u)_{u=1}^v \right)_{v=2}^{p} 
 &= (a_1, a_2), (a_1, a_2, a_3), (a_1, a_2, a_3, a_4), \ldots, (a_1, a_2, \ldots, a_{p}).
\end{align*}
Since  for $v \in [2, p]$ and $u \in [1, v]$, each $b_{v-u+1} a_u$ is unique, 
the desired result follows.
\end{proof}

\begin{lemma}\label{lem:lz-t-p}
$z(T_p) = 3p^2/2 + 3p/2$.
\end{lemma}

\begin{proof}
By \cref{obs:uniq-b-j-a-k}, the Lempel-Ziv factorization of $T_p$ contains a phrase boundary 
between each consecutive $G_j$ and $G_{j+1}$ (blue boundaries on the left of \cref{fig:lz-t-p}).

We next examine the phrases inside each $G_j$ for $j \in [2, m_p]$.
For $s \in [2, p]$, let $m_s = \sum_{i=1}^s i$.
Then each $j \in [2, m_p]$ can be written as
$j = m_s - d$ for some $s \in [2, p]$ and $d \in [0, s-1]$. 
Now, observe that, $G_j$ yields the following phrases depending on $d$ and $j$:
\[
G_j \text{~yields~}
\begin{cases}
G_j[1 \ldots s-d],\; G_j[s-d+1 \ldots |G_j|-1],\; G_j[|G_j|]
& \text{if $d>0$,}\\
G_j[1],\; G_j[2 \ldots |G_j|-1],\; G_j[|G_j|]
& \text{if $d=0$ and $j < m_p$,}\\
G_j[1],\; G_j[2 \ldots |G_j|]
& \text{if $d=0$ and $j = m_p$.}
\end{cases}
\]
In \cref{fig:lz-t-p} (left), for each $G_j$, 
the first boundary is red when $d>0$ and green when $d = 0$;
the second boundary is shown in orange for $j < m_p$.

Finally, since $z(G_1)=4$,
it follows that $z(T_p) = 4 + 3(m_p-2)+2 = 3p^2/2 + 3p/2$.
\end{proof}

\begin{lemma}\label{lem:lz-t-p-r}
$z(T_p^R) = p^2/2 + 5p/2 - 2$.
\end{lemma}

\begin{proof}
First observe that, in $T_p$, 
each $G_j$ occurs in $G_k$ only when $k \geq j$, 
and does not occur in any $G_{k'}$ with $k' < j$. 
Consequently, in $T_p^R = G_{m_p}^R \cdots G_1^R$,
each $G_j^R$ for $j \in [1, m_p-1]$ has an occurrence in $G_{k}^R$ for $k \geq j$.
Hence, by \cref{obs:uniq-b-j-a-k}, each $G_j^R$ is parsed as a separate phrase in the Lempel-Ziv factorization of $G_{m_p-1}^R \cdots G_1^R$ (blue boundaries on the right of \cref{fig:lz-t-p}). 
Further, since $G_{m_p}^R = (\mathcal{A}_p \cdot  x \cdot  \mathcal{B}_p)^R$,
we have $z(G_{m_p}^R) = z(\mathcal{B}_p^R \cdot  x \cdot  \mathcal{A}_p^R) = (2p-1) + 1 +(2p-1) = 4p-1$ by \cref{obs:lz-A-p-B-p}.
Therefore, $z(T_p^R) = z(G_{m_p}^R) + (m_p-1) = p^2/2 + 5p/2 - 2$.
\end{proof}

\begin{proposition}
There exists a string family where $\liminf_{n  \to \infty} z(w^R)/z(w) = 3$.
The same result holds for $z_{no}$. 
\end{proposition}

\noindent
In contrast, we do not obtain an analogous result for $z_e$; 
exploring whether such a family exists for $z_e$ is left for future work.
We also observe that both $z(T_p)$ and $z(T_p^R)$ are in  $O(\sqrt{|T_p|})$.
Hence, the multiplicative sensitivity does not imply a linear additive sensitivity.
To obtain a linear additive gap, we consider a different string family in the next subsection.

\subsection{Additive Sensitivity}

Constantinescu~\cite{ConstantinescuThesis} showed that the size $z$ of the Lempel-Ziv factorization of a string can increase by $\Omega(\sqrt{n})$ when reversing the string\footnote{This holds for an asymptotically equivalent LZ variant, which adds a trailing character at each phrase}. We improve this lower-bound to $\Omega(n)$, and show that this new result also holds for the sensitivity of the Lempel-Ziv parsing without overlaps and the Lempel-Ziv-end variants.
The infinite family of strings that we rely on is defined as follows. 
(The same family will be reused in \cref{sec:add-sen-lex-parse}.)

\begin{definition}[$w_\sigma$]
Let $\sigma\ge2$ be an \emph{even} integer and $\Sigma = \{a_1,\dots,a_\sigma\}$
such that $a_1 < \cdots < a_\sigma$.
We then define 
$
w_\sigma = ( \prod_{i=1}^{\sigma-1} a_i a_{i+1} ) ( \prod_{i=1}^{\sigma} a_i ),
$
noting that $n = |w_\sigma| = 3\sigma-2$.
\end{definition}

\begin{lemma}\label{lem:z-w-sigma}
Let $\sigma\ge2$ be an even integer. Then, $z(w_\sigma) = 2\sigma+\sigma/2-2$.
\end{lemma}

\begin{proof}
The Lempel-Ziv factorization of $w_\sigma$ is 
$
\left(
(a_i, a_{i+1})_{i=1}^{\sigma-1},
(a_{2j-1} a_{2j})_{j=1}^{\sigma/2}
\right),
$
where the first $2(\sigma-1)$ phrases are
$(a_i, a_{i+1})_{i=1}^{\sigma-1} 
= (a_1, a_2, a_2, a_3, \ldots, a_{\sigma-1}, a_{\sigma})$,
and the remaining $\sigma/2$ phrases are
$(a_{2j-1} a_{2j})_{j=1}^{\sigma/2}
= (a_1 a_2, a_3 a_4, \ldots, a_{\sigma-1} a_{\sigma})$.
Thus, $z(w_\sigma) = 2\sigma + \sigma/2 - 2$.
\end{proof}

\begin{lemma}\label{lem:z-w-sigma-r}
Let $\sigma\ge2$ be an even integer. Then, $z(w_\sigma^R) = 2\sigma-1$.
\end{lemma}

\begin{proof}
First observe that $w_\sigma^R = ( \prod_{i=0}^{\sigma-1} a_{\sigma-i} ) (\prod_{i=0}^{\sigma-2} a_{\sigma-i} a_{\sigma-i-1})$.
Then, the Lempel-Ziv factorization of $w_\sigma^R$ is 
$
\left(
(a_{\sigma-i})_{i=0}^{\sigma-1}, 
(a_{\sigma-i} a_{\sigma-i-1})_{i=0}^{\sigma-2}
\right)
$. So we have $z(w_\sigma^R) =   2\sigma-1$.
\end{proof}

The next proposition follows from \cref{lem:z-w-sigma} and \cref{lem:z-w-sigma-r}.
The other Lempel-Ziv variants yield the exact same factorizations on $w_\sigma$ and $w_\sigma^R$. 

\begin{proposition}\label{prop:add-sen-z}
There exists an infinite string family where $z(w^R)-z(w)=\frac{n+2}{6}-1=\Theta(n)$. 
The same result holds for $z_{no}, z_e$ and $z_{end}$.
\end{proposition}

\section{Sensitivity of the Smallest Lexicographic Parsing}\label{sec:sen-lex-parse}

\subsection{Multiplicative Sensitivity}

In this subsection we prove that, similarly to the RLBWT and its variants~\cite{BCLR2025,GILPST21},
the size of the lex-parse of a string can increase by a logarithmic factor upon reversing the string. To prove this result, we rely on a family of words constructed from the Fibonacci words.

\begin{definition}[$F_k$ and $C_k$]
Let $F_k$ denote the $k^{\text{th}}$ \emph{Fibonacci word} where
$F_1 = \b$, $F_2 = \a$, and $F_k = F_{k-1} F_{k-2}$ for each $k \geq 3$.
We denote $f_k = |F_k|$ the lengths of these words. Moreover,
we let $C_k = F_k[1 \ldots f_k-2]$ be the $k^{\text{th}}$ \emph{central word}.
\end{definition}

We summarize in the following lemma some known properties of Fibonacci words.

\begin{lemma}[\cite{GILPST21,NOP2021}]\label{lem:fib-c-properties}
For each  $k \geq 6$, 
\begin{enumerate}[(i)]
\item\label{property:palindrome} $C_k$ is a palindrome.
\item\label{property:c-rev-fac} $C_k = F_{k-2} \ F_{k-3} \cdots  F_3 \ F_2$.
\item\label{property:acb-lyndon} $\a C_k \b$ is Lyndon.
\item\label{property:f-k-2-occ} $F_{k-2}$ occurs in $F_{k}$ exactly three times, starting at positions 1, $f_{k-2}+1$, and $f_{k-1}+1$.
\item\label{property:c-ab-ba} 
$
F_{k}=C_k \cdot
\begin{cases}
\a\b & \text{if $k$ is odd}, \\
\b\a & \text{if $k$ is even};
\end{cases}
\quad
F_{k-2}F_{k-1}=C_k \cdot
\begin{cases}
\b\a & \text{if $k$ is odd}, \\
\a\b & \text{if $k$ is even}.
\end{cases}
$
\end{enumerate}
\end{lemma}

We also use the following well-known property of Lyndon words.

\begin{lemma}[\cite{chen1958free}]\label{lem:uv-lyndon}
Consider Lyndon words $u$ and $v$ such that $u < v$.
Then, $uv$ is Lyndon.
\end{lemma}

\begin{figure}[t]
\centering
\includegraphics[width=0.6\linewidth]{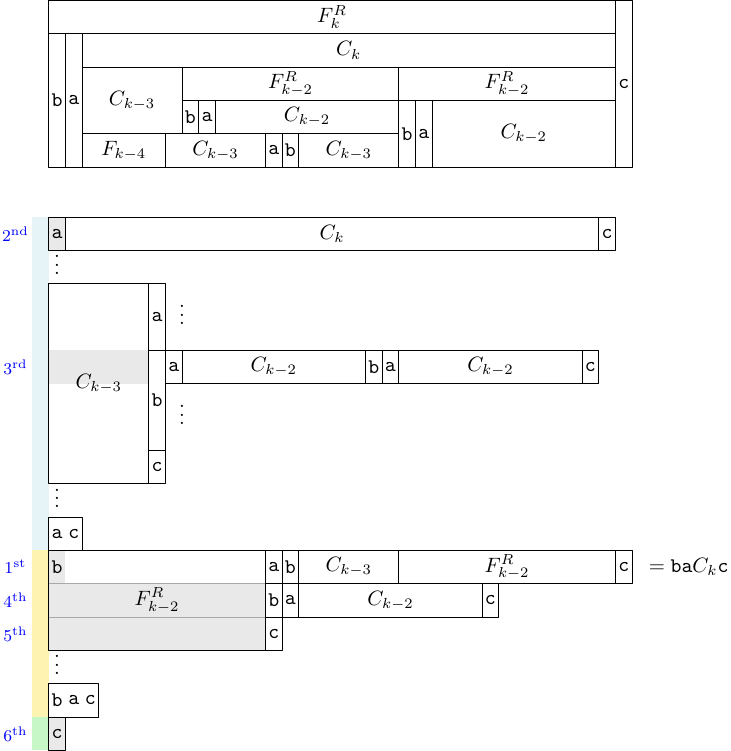}
\caption{
Illustration of \cref{lem:v-odd-fib-rev-c}.
Top: several factorizations of $w = F_k^R \c$ for odd $k$.
Bottom: the sorted suffixes of $w$ 
(suffixes starting with $\a$, $\b$, and $\c$ are shown in three colors on the left;
the ordinals indicate the order of the six phrases in the lex-parse of $w$ highlighted in gray.)
}
\label{fig:fib-rev-c}
\end{figure}

The main result we prove in this section is the following.

\begin{proposition}There exists a string family where $v(w^R)/v(w)=\Theta(\log n)$.
\end{proposition}

This lower-bound is asymptotically tight, as in \cref{sec:preliminaries} we showed $v(w^R)/v(w)=O(\log n)$. The string family for which this proposition holds is composed of the reverses of strings $\c F_k=\c C_k\a\b$ for odd $k$. It is known that $v = \Theta(\log n)$ on odd Fibonacci words $F_k=C_k\a\b$~\cite{NOP2021}, and prepending $\c$ to this string only increases $v$ by 1, as $\c$ is a unique symbol that does not interfere with the relative order of the other suffixes of $F_k$. 

\begin{lemma}\label{lem:v-odd}
For each odd $k$, it holds that $v(\c F_k)=\Theta(\log n)$.
\end{lemma}

Therefore, it remains to prove that $v=O(1)$ on the strings $(\c F_k)^R$ for odd $k$. 
Specifically, we show that $v \le 6$ on these strings.

\begin{lemma}\label{lem:v-odd-fib-rev-c}
For each odd $k \geq 9$, $v((\c F_k)^R) = 6$.
\end{lemma}

\begin{proof}
See \cref{fig:fib-rev-c}.
Let $w = (\c F_k)^R = (\c C_k \a \b)^R =\b \a C_k \c$.
First, since $\a C_k \b$ is Lyndon by (\ref{property:acb-lyndon}) of \cref{lem:fib-c-properties}, 
$\a C_k \b < C_k[j \ldots |C_k|] \b$ holds for any $j \in [1, |C_k|]$.
This implies that $\a C_k  < C_k[j \ldots |C_k|] $ and
 $\a C_k \c <  C_k[j \ldots |C_k|] \c$.
We infer that $w[2 \ldots |w|] = \a C_k \c$ is the smallest suffix of $w$ that begins with $\a$.
Further,   
we know that $w[1 \ldots |w|] = \b \a C_k \c$ is the smallest suffix of $w$ that begins with $\b$.
Thus, the first two phrases are $\b$ and $\a$, respectively. 

We now prove that, starting at position $3$ in $w$, the third phrase is $C_{k-3}$.
By repeatedly applying the definition $F_k = F_{k-1} F_{k-2}$, we have the following factorizations of $F_k$:
$F_k 
= F_{k-2} F_{k-3} F_{k-3} F_{k-4} 
 = F_{k-2} F_{k-3} F_{k-4} F_{k-5} F_{k-4} = F_{k-2} F_{k-2} F_{k-5} F_{k-4}$.
Since $k$ is odd, by (\ref{property:c-ab-ba}) of \cref{lem:fib-c-properties}, 
$F_{k-5}F_{k-4} = C_{k-3}\a\b$, so $F_k=F_{k-2} F_{k-2} C_{k-3} \a\b$.
By (\ref{property:palindrome}) of \cref{lem:fib-c-properties}, $C_{k-3}$ is a palindrome; 
it follows that $F_k^R = \b\a C_{k-3} F_{k-2}^R F_{k-2}^R$.
Thus, in $w = F_k^R \c$, the factor immediately following the first occurrence of $C_{k-3}$ is 
$
F_{k-2}^R F_{k-2}^R \c 
= \b\a C_{k-2}\b\a C_{k-2}\c
$.
In particular, 
by (\ref{property:acb-lyndon}) of \cref{lem:fib-c-properties}, $\a C_{k-2} \b$ is Lyndon.
From the first paragraph,
we also have that $\a C_{k-2}\c$ is Lyndon.
Since $\a C_{k-2} \b < \a C_{k-2}\c$,
\cref{lem:uv-lyndon} implies that  $\a C_{k-2}\b\a C_{k-2}\c$ is Lyndon.
Hence, $w[3 \ldots |w|] = C_{k-3} \b\a C_{k-2}\b\a C_{k-2}\c$
is the smallest suffix of $w$ that begins with $C_{k-3} \b$.
Next, we derive the following identities using (\ref{property:c-rev-fac}) and (\ref{property:c-ab-ba}) of \cref{lem:fib-c-properties}:
$C_{k-3}  = F_{k-5} F_{k-6} F_{k-7} \cdots F_3 F_2 = F_{k-5} F_{k-6} C_{k-5} = F_{k-4} C_{k-5}$;
$F_{k-2}^R = F_{k-4}^R F_{k-3}^R = F_{k-4}^R \a \b C_{k-3}$; 
$C_{k-5} F_{k-4}^R  = C_{k-5} \b \a C_{k-4} = F_{k-5} C_{k-4} = F_{k-5}  F_{k-6} \cdots F_3 F_2 = C_{k-3}$; and
$
  C_{k-3} F_{k-2}^R
= F_{k-4} C_{k-5}  F_{k-4}^R \a \b C_{k-3}
= F_{k-4} C_{k-3} \a \b C_{k-3}
$.
These identities imply  a factorization of $w$ as $\b \a F_{k-4} C_{k-3} \a \b  C_{k-3} \b \a C_{k-2}$.
Thus, there exists an occurrence of $C_{k-3} \a$ in $w$,
and the lexicographic predecessor of $w[3 \ldots |w|]$ is a suffix of $w$ beginning with $C_{k-3}\a$.
Therefore,  the longest common prefix between these two suffixes is $C_{k-3}$, which is the third phrase.

By (\ref{property:f-k-2-occ}) of \cref{lem:fib-c-properties}, there are exactly three occurrences of $F_{k-2}^R$ in $F_{k}^R$, each followed by $\a$, $\b$, $\c$, and the last two of these occurrences are not the one followed by $\a$, it follows that the next two phrases are both $F_{k-2}^R$.
Finally, the last phrase is $\c$.
Therefore, $\LEX(w) = (\b,  \a,   C_{k-3},   F_{k-2}^R,   F_{k-2}^R,  \c)$.
\end{proof}

\subsection{Additive Sensitivity}\label{sec:add-sen-lex-parse}

In this subsection we show that $v$ has $\Theta(n)$ additive sensitivity to the reverse operation.

\begin{lemma}\label{lem:v-w-sigma}
$v(w_\sigma) =  2\sigma + \sigma/2 - 2$.
\end{lemma}

\begin{proof}
Among the suffixes of $w_\sigma = a_1 a_2 a_2 a_3 \cdots a_{\sigma-1} a_{\sigma} \cdot a_1 \cdots a_{\sigma}$,
we observe the following.
\begin{itemize}
\item 
Two suffixes start with $a_1$; 
in text position order, they start with $a_1a_2a_2$ and $a_1a_2a_3$.
\item
For each $i \in [2, \sigma-2]$, three suffixes start with $a_i$; 
in text position order, they start with $a_ia_ia_{i+1}$, $a_ia_{i+1}a_{i+1}$, and $a_ia_{i+1}a_{i+2}$. 
\item
Three suffixes start with $a_{\sigma-1}$;
in text position order, they start with $a_{\sigma-1}a_{\sigma-1}a_{\sigma}$, $a_{\sigma-1}a_{\sigma}a_{1}$, and $a_{\sigma-1}a_{\sigma}$.
\item
Two suffixes start with $a_{\sigma}$; 
in text position order, they start with $a_\sigma a_1$ and $a_\sigma$.
\end{itemize}
With these observations, and by inspecting the longest common prefixes of consecutive suffixes (see \cref{fig:lex-parse-w-sigma}), we obtain that
$\LEX(w_\sigma) =
\left( 
(a_i, a_{i+1})_{i=1}^{\sigma-2}, 
a_{\sigma-1} a_{\sigma},
(a_{2j-1} a_{2j})_{j=1}^{\sigma/2-1},
a_{\sigma-1}, 
a_{\sigma}
\right),
$
where 
$(a_i, a_{i+1})_{i=1}^{\sigma-2}
= (a_1, a_2,  \ldots, a_{\sigma-2}, a_{\sigma-1})$
and 
$(a_{2j-1} a_{2j})_{j=1}^{\sigma/2-1}
= (a_1 a_2,  \ldots, a_{\sigma-3} a_{\sigma-2})$.
So we have  $v(w_\sigma) =  2(\sigma-2) + 1 + (\sigma/2-1) + 2 = 2\sigma + \sigma/2 - 2$.
\end{proof}

\begin{figure}[t]
    \centering
    \includegraphics[width=0.7\linewidth]{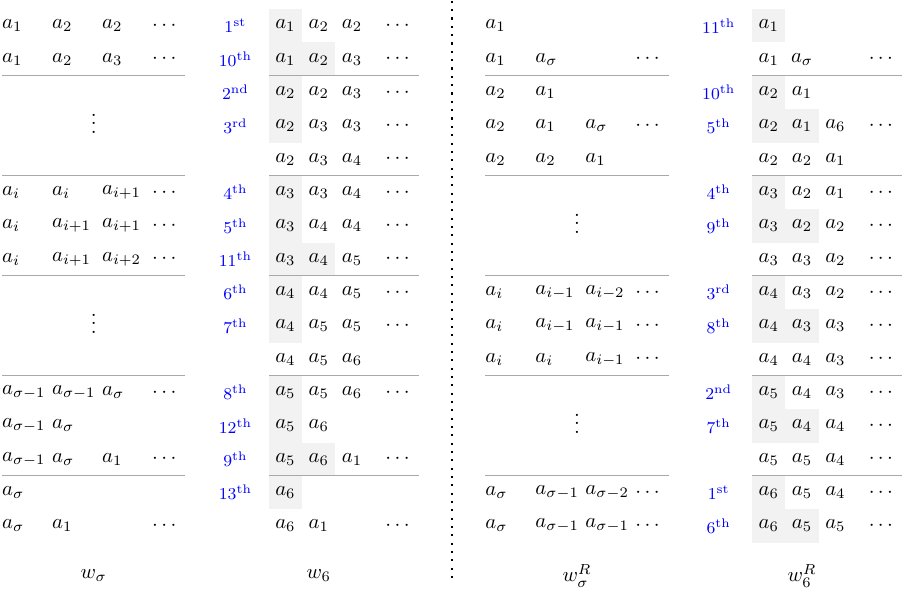}
    \caption{Illustration of \cref{lem:v-w-sigma} and \cref{lem:v-w-sigma-r}: 
    prefixes of sorted suffixes of $w_\sigma$ (left) and $w_\sigma^R$ (right). 
    The ordinals indicate the order of the highlighted phrases for $\sigma=6$ 
    where $w_6 = a_1 a_2 a_2 a_3 a_3 a_4 a_4 a_5 a_5 a_6 a_1 a_2 a_3 a_4 a_5 a_6$.
    }
    \label{fig:lex-parse-w-sigma}
\end{figure}

\begin{lemma}\label{lem:v-w-sigma-r}
$v(w_\sigma^R) =  2\sigma-1$.
\end{lemma}

\begin{proof}
By symmetric observations of the proof of \cref{lem:v-w-sigma}, 
it can be verified that $\LEX(w_\sigma^R)=
\left(
(a_{\sigma-i})_{i=0}^{\sigma-3}, \
a_2 a_1, \
(a_{\sigma-i} a_{\sigma-i-1})_{i=0}^{\sigma-3}, \
a_2, \
a_1
\right),
$
where
$(a_{\sigma-i})_{i=0}^{\sigma-3} = (a_\sigma,  \ldots, a_3)$
and 
$(a_{\sigma-i} a_{\sigma-i-1})_{i=0}^{\sigma-3} 
= (a_\sigma a_{\sigma-1},  \ldots, a_3 a_2)$.
So we have $v(w_\sigma^R) = 2(\sigma-2)+3 = 2\sigma-1$.
\end{proof}

Note that although $z(w_\sigma) = v(w_\sigma)$ and $z(w_\sigma^R) = v(w_\sigma^R)$,
the factorizations are different.

\begin{proposition}\label{prop:add-sen-lex-parse}
There exists a string family where $v(w^R)-v(w)=  \frac{n+2}{6}-1=\Theta(n)$.
\end{proposition}

\section{Conclusions}\label{sec:conclusions}

We have presented several results on the sensitivity of repetitiveness measures to string reversal. We have described an infinite string family exhibiting a $\Theta(n)$ additive increase in the number $r$ of BWT runs under reversal. We gave analogous results for the size of the Lempel-Ziv parse, and the lex-parse. We also described a string family where the Lempel-Ziv parse ratio $z(w^R)/z(w)$ can approach 3 when applying the reverse operation. Finally, we described a string family where the size $v$ of the lex-parse can increase by an $\Omega(\log n)$ factor after reversing the string, and this increase is asymptotically tight.

Many open problems remain in this general direction. Two central questions are whether $z(w^R)/z(w)=O(1)$ holds and whether $r(w^R)/r(w)=O(\log n)$ holds. Another related direction is the asymptotic relation between $z$ and $v$. While $v=o(z)$ for even Fibonacci words, it is unknown if $v = O(z)$ in general. If indeed $z$ has $O(1)$ multiplicative sensitivity against the reverse operation, as conjectured, this may offer further insight into understanding how $z$ and $v$ relate. Beyond these, several bounds for $\chi$, $e$ (as can be seen in \cref{table:sensitivity_reverse}), and other measures can be improved. It is worth noting that, although our lower bounds are asymptotically tight, determining the exact tight constants remains an interesting combinatorial question. There may also be conditional lower bounds worth exploring---for instance, as a function of the alphabet size. 
For an extended version of this paper, we also plan to explore the impact of string reversal on other variants of the Lempel-Ziv parsing, including the family of Lempel-Ziv height-bounded (LZHB) parsings~\cite{LZHB}.

Overall, our results elucidate the limitations of many repetitiveness measures with respect to string reversal. We hope this work brings renewed attention to this string operation and motivates progress on these long-standing questions. 

\clearpage

\bibliography{bibliography}

\appendix

\section{\texorpdfstring{$\Omega(\log n)$}{} Multiplicative Sensitivity of \texorpdfstring{$\rdol$}{}}\label{sec:mul-sen-rdol}

\begin{proposition}
There is a string family where $\rdol(w^R)/\rdol(w)=\Omega(\log n)$.
\end{proposition}

\begin{proof}
Let $F_k$ be an even Fibonacci word, so it ends with an $\a$. Then, $F_k\a$ is a \emph{Fibonacci-plus word}~\cite{GILPST21}. It is known that for Fibonacci-plus words, it holds that $r(F_k\a) = 4$ and $r((F_k\a)^R)=\Omega(\log n)$~\cite{GILPST21}. These properties hold for even and odd $k$, but we focus on even $k$. A crucial observation we use is that in both $F_k\a$ and its reverse, there is a unique circular occurrence of the substring $\a\a\a$, which is the smallest substring of its length.

Let $x\a\a\a$ and $y\a\a\a$ be the smallest prefixes of the infinite powers of the rotations $w_i$ and $w_j$ of $w=F_k\a$, respectively, such that $|x|<|y|$ and the occurrence of $\a\a\a$ is in both cases unique. Observe that because $\a\a\a$ is unique and it is the smallest circular substring of length $3$ in $w$, it holds $w_i > w_j$ if and only if $x > y$. After appending $\dol$ to $F_k\a$, we obtain that those prefixes become $x\a\a\dol\a$ and $y\a\a\dol\a$, and the same order is preserved, with the exception of the new rotation starting with $\dol$ which is now the first rotation, and the rotation starting with $\a\dol\a\b$, which is the second rotation (these are the only cases where $\dol$ occurs outside $\a\a\dol\a$).

Similarly, let $x\a\a\a$ and $y\a\a\a$ satisfy the same conditions with respect to rotations $w^R_i$ and $w^R_j$ of $w^R=\a F_k^R$. Again, it holds $w^R_i > w^R_j$ if and only if $x > y$. After appending $\dol$ to $\a F_k^R$, we obtain that those prefixes become $x\a\dol\a\a$ and $y\a\dol\a\a$, and the same order of the rotations is preserved, with the exception of the new rotation starting with $\dol$ which is the first one.

That is, the differences $|r(F_k\a)-\rdol(F_k\a)|$ and $|r((F_k\a)^R)-\rdol((F_k\a)^R)|$ are $O(1)$, and thus, $\rdol((F_k\a)^R)/\rdol(F_k\a)=\Omega(\log n)$.
\end{proof}

\section{\texorpdfstring{$\Omega(n)$}{} Additive Sensitivity of \texorpdfstring{$e$}{}}\label{sec:add-sen-e}

A substring $x$ of $w$ is a \emph{maximal repeat} on $w$ if it occurs at least twice in $w$ and at least one of the following conditions holds: (i) there exist distinct symbols $a$ and $b$ such that both $xa$ and $xb$ occur in $w$; (ii) there exist distinct symbols $a$ and $b$ such that both $ax$ and $bx$ occur in $w$; (iii) $x$ is a prefix of $w$; (iv) $x$ is a suffix of $w$. We denote by $M(w)$ the set of maximal repeats of $w$. The set of \emph{right-extensions} of $w$ is denoted by $E_r(w)=\{xa\,|\,x \text{ is right-maximal in } w \}$. 

The \emph{compact directed acyclic word graph (CDAWG)} is a widely used indexing data structure, and its number of edges $e$ is considered a measure of repetitiveness~\cite{NavSurveyACM}. For our purposes, we use a known characterization $e=|E_r(w)|$~\cite{IK2025arxiv}. 

\begin{proposition}
    There exists a string family where $e(w^R)-e(w) = n-2 = \Theta(n)$.
\end{proposition}

\begin{proof}
Such a family is composed of the strings $w = \a\b^{n-1}$. Observe that $M(w)=\bigcup_{i=0}^{n-2}\{\b^i\} = M(w^R)$. The right-extensions of $w$ are $E_r(w)=\{\a\} \cup \bigcup_{i=1}^{n-1} \{\b^i\}$, and therefore $e(w)=n$. On the other hand, the right extensions of $w^R=\b^{n-1}\a$ are $E_r(w^R)= \bigcup_{i=0}^{n-2} \{\b^i\a,\b^i\b\}$, and therefore $e(w^R)= 2(n-1)$. Thus, the claim holds.
\end{proof}

\end{document}